\documentclass[12pt,letterpaper]{article}

\usepackage{amsmath,amssymb,amsthm,amsfonts}
\usepackage{accents}
\usepackage{caption}
\usepackage{comment}
\usepackage[roman,full]{complexity}
\usepackage{enumerate}
\usepackage{fancyhdr}
\usepackage{float}
\usepackage{fullpage}
\usepackage{graphicx}
\usepackage{hyperref}
\usepackage{parskip}
\usepackage{todonotes}
\usepackage[square,numbers]{natbib}
\usepackage{dsfont}

\renewcommand{\E}{\mathbb{E}}

\theoremstyle{definition}
\newtheorem{theorem}{Theorem}
\newtheorem{proposition}{Proposition}
\newtheorem{corollary}{Corollary}
\newtheorem{lemma}{Lemma}
\newtheorem{definition}{Definition}
\newtheorem{remark}{Remark}

\allowdisplaybreaks

\title{Bounds on the Total Coefficient Size of Nullstellensatz Proofs of the Pigeonhole Principle and the Ordering Principle}
\author{Aaron Potechin and Aaron Zhang}
\date{}

\begin{document}
\maketitle
\abstract{In this paper, we investigate the total coefficient size of Nullstellensatz proofs. We show that Nullstellensatz proofs of the pigeonhole principle on $n$ pigeons require total coefficient size $2^{\Omega(n)}$ and that there exist Nullstellensatz proofs of the ordering principle on $n$ elements with total coefficient size $2^n - n$.}\\
\\
\textbf{Acknowledgement:} This research was supported by NSF grant CCF-2008920 and NDSEG fellowship F-9422254702.
\section{Introduction}
Given a system $\{p_i = 0: i \in [m]\}$ of polynomial equations over an algebraically closed field, a Nullstellensatz proof of infeasibility is an equality of the form $1 = \sum_{i=1}^{m}{p_i{q_i}}$ for some polynomials $\{q_i = 0: i \in [m]\}$. Hilbert's Nullstellensatz\footnote{Actually, this is the weak form of Hilbert's Nullstellensatz. Hilbert's Nullstellensatz actually says that given polynomials $p_1,\ldots,p_m$ and another polynomial $p$, if $p(x) = 0$ for all $x$ such that $p_i(x) = 0$ for each $i \in [m]$ then there exists a natural number $r$ such that $p^r$ is in the ideal generated by $p_1,\ldots,p_m$. 
} says that the Nullstellensatz proof system is complete, i.e. a system of polynomial equations has no solutions over an algebraically closed field if and only if there is a Nullstellensatz proof of infeasibility. However, Hilbert's Nullstellensatz does not give any bounds on the degree or size needed for Nullstellensatz proofs.



The degree of Nullstellensatz proofs has been extensively studied. Grete Hermann showed a doubly exponential degree upper bound for the ideal membership problem \cite{gretehermann} which implies the same upper bound for Nullstellensatz proofs. Several decades later, W. Dale Brownawell gave an exponential upper bound on the degree required for Nullstellensatz proofs over algebraically closed fields of characterisic zero \cite{10.2307/1971361}. A year later, J{\'a}nos Koll{\'a}r showed that this result holds for all algebraically closed fields \cite{kollar1988sharp}. 

For specific problems, the degree of Nullstellensatz proofs can be analyzed using designs \cite{DBLP:conf/dimacs/Buss96}. Using designs, Nullstellensatz degree lower bounds have been shown for many problems including the pigeonhole principle, the induction principle, the housesitting principle, and the mod $m$ matching principles \cite{365714, 10.1006/jcss.1998.1575, 10.1007/BF01294258, 507685, 10.1145/237814.237860}. More recent work showed that there is a close connection between Nullstellensatz degree and reversible pebbling games
\cite{derezende_et_al:LIPIcs:2019:10840} and that lower bounds on Nullstellensatz degree can be lifted to lower bounds on monotone span programs, monotone comparator circuits, and monotone switching networks \cite{10.1145/3188745.3188914}. 

For analyzing the size of Nullstellensatz proofs, a powerful technique is the size-degree tradeoff  showed by Russell Impagliazzo, Pavel Pudl\'{a}k, and Ji\v{r}\'{\i} Sgall for polynomial calculus \cite{10.1007/s000370050024}. This tradeoff says that if there is a size $S$ polynomial calculus proof then there is a polynomial calculus proof of degree $O(\sqrt{n\log{S}})$. Thus, if we have an $\Omega(n)$ degree lower bound for polynomial calculus, this implies a $2^{\Omega(n)}$ size lower bound for polynomial calculus (which also holds for Nullstellensatz as Nullstellensatz is a weaker proof system). 
However, the size-degree tradeoff does not give any size lower bound when the degree is $O(\sqrt{n})$ and we know of very few other techniques for analyzing the size of Nullstellensatz proofs.

In this paper, we instead investigate the total coefficient size of Nullstellensatz proofs. We have two reasons for this. First, total coefficient size is interesting in its own right and to the best of our knowledge, it has not yet been explored. Second, total coefficient size may give insight into proof size in settings where we cannot apply the size-degree tradeoff and thus do not have good size lower bounds. 
\begin{remark}
Note that Nullstellensatz size lower bounds do not imply total coefficient size lower bounds because we could have a proof with many monomials but a small coefficient on each monomial. Thus, the exponential size lower bounds for the pigeonhole principle from Razborov's $\Omega(n)$ degree lower bound for polynomial calculus \cite{razborov1998lower} and the size-degree tradeoff \cite{10.1007/s000370050024} do not imply total coefficient size lower bounds for the pigeonhole principle.
\end{remark}


\subsection{Our results}
In this paper, we consider two principles, the pigeonhole principle and the ordering principle. We show an exponential lower bound on the total coefficient size of Nullstellensatz proofs of the pigeonhole principle and we show an exponential upper bound on the total coefficient size of Nullstellensatz proofs of the ordering principle. More precisely, we show the following bounds.
\begin{theorem}\label{thm:pigeonholelowerbound}
For all $n \geq 2$, any Nullstellensatz proof of the pigeonhole principle with $n$ pigeons and $n-1$ holes has total coefficient size $\Omega\left(n^{\frac{3}{4}}\left(\frac{2}{\sqrt{e}}\right)^{n}\right)$.
\end{theorem}
\begin{theorem}
For all $n \geq 3$, there is a Nullstellensatz proof of the ordering principle on $n$ elements with size and total coefficient size $2^{n} - n$.
\end{theorem}

After showing these bounds, we discuss total coefficient size for stronger proof systems. We observe that if we consider a stronger proof system which we call resolution-like proofs, our lower bound proof for the pigeonhole principle no longer works. We also observe that even though resolution is a dynamic proof system, the $O(n^3)$ size resolution proof of the ordering principle found by Gunnar St{\aa}lmark \cite{staalmarck1996short} can be captured by a one line sum of squares proof.
\section{Nullstellensatz total coefficient size}\label{preliminaries}
We start by defining total coefficient size for Nullstellensatz proofs and describing a linear program for finding the minimum total coefficient size of a Nullstellensatz proof. \begin{definition}
Given a polynomial $f$, we define the total coefficient size $T(f)$ of $f$ to be the sum of the magnitudes of the coefficients of $f$. For example, if $f(x,y,z) = 2{x^2}y - 3xyz + 5z^5$ then $T(f) = 2 + 3 + 5 = 10$.
\end{definition}
\begin{definition}
Given a system $\{p_i = 0: i \in [m]\}$ of $m$ polynomial equations, a Nullstellensatz proof of infeasibility is an equality of the form 
\[
1 = \sum_{i=1}^{m}{p_i{q_i}} 
\]
for some polynomials $\{q_i: i \in [m]\}$. We define the total coefficient size of such a Nullstellensatz proof to be $\sum_{i=1}^{m}{T(q_i)}$.
\end{definition}
The following terminology will be useful.
\begin{definition}
Given a system $\{p_i = 0: i \in [m]\}$ of polynomial equations, we call each of the equations $p_i = 0$ an axiom. For each axiom $s_i = 0$, we define a weakening of this axiom to be an equation of the form $rp_i = 0$ for some monomial $r$.
\end{definition}
\begin{remark}
We do not include the total coefficient size of $p_i$ in the total coefficient size of the proof as we want to focus on the complexity of the proof as opposed to the complexity of the axioms. That said, in this paper we only consider systems of polynomial equations where each $p_i$ is a monomial, so this choice does not matter.
\end{remark}
The minimum total coefficient size of a Nullstellensatz proof can be found using the following linear program. In general, this linear program will have infinite size, but as we discuss below, it has finite size when the variables are Boolean.
\begin{enumerate}
    \item[] Primal: Minimize $\sum_{i=1}^{m}{T(q_i)}$ subject to $\sum_{i=1}^{m}{{p_i}{q_i}} = 1$. More precisely, writing $q_i = \sum_{\text{monomials } r}{c_{ir}r}$, we minimize $\sum_{i=1}^{m}{\sum_{\text{monomials } r}{b_{ir}}}$ subject to the constraints that 
    \begin{enumerate}
        \item[1.] $b_{ir} \geq -c_{ir}$ and $b_{ir} \geq c_{ir}$ for all $i \in [m]$ and monomials $r$.
        \item[2.] $\sum_{i=1}^{m}{\sum_{\text{monomials } r}{c_{ir}{r}p_i}} = 1$
    \end{enumerate}
    \item[] Dual: Maximize $D(1)$ subject to the constraints that
    \begin{enumerate}
        \item[1.] $D$ is a linear map from polynomials to $\mathbb{R}$.
        \item[2.] For each $i \in [m]$ and each monomial $r$, $|D(rp_i)| \leq 1$.
    \end{enumerate}
\end{enumerate}
Weak duality, which is what we need for our lower bound on the pigeonhole principle, can be seen directly as follows.
\begin{proposition}
If $D$ is a linear map from polynomials to $\mathbb{R}$ such that $|D(rp_i)| \leq 1$ for all $i \in [m]$ and all monomials $r$ then any Nullstellensatz proof of infeasibility has total coefficient size at least $D(1)$.
\end{proposition}
\begin{proof}
Given a Nullstellensatz proof $1 = \sum_{i=1}^{m}{{p_i}{q_i}}$, applying $D$ to it gives
\[
D(1) = \sum_{i=1}^{m}{D({p_i}{q_i})} \leq \sum_{i=1}^{m}{T(q_i)}
\]
\end{proof}
\subsection{Special case: Boolean variables}
In this paper, we only consider problems where all of our variables are Boolean, so we make specific definitions for this case. In particular, we allow monomials to contain terms of the form $(1-x_i)$ as well as $x_i$ and we allow the Boolean axioms $x_i^2 = x_i$ to be used for free. We also observe that we can define a linear map $D$ from polynomials to $\mathbb{R}$ by assigning a value $D(x)$ to each input $x$.
\begin{definition}
Given Boolean variables $x_1,\ldots,x_N$ where we have that $x_i = 1$ if $x_i$ is true and $x_i = 0$ if $x_i$ is false, we define a monomial to be a product of the form $\left(\prod_{i \in S}{x_i}\right)\left(\prod_{j \in T}{(1 - x_j)}\right)$ for some disjoint subsets $S,T$ of $[N]$.
\end{definition}
\begin{definition}
Given a Boolean variable $x$, we use $\bar{x}$ as shorthand for the negation $1-x$ of $x$.
\end{definition}
\begin{definition}
Given a set of polynomial equations $\{p_i = 0: i \in [m]\}$ together with Boolean axioms $\{x_j^2 - x_j = 0: j \in [N]\}$, we define the total coefficient size of a Nullstellensatz proof
\[
1 = \sum_{i = 1}^{m}{{p_i}{q_i}} + \sum_{j = 1}^{N}{{g_j}(x_j^2 - x_j)}
\]
to be $\sum_{i=1}^{m}{T(q_i)}$. In other words, we allow the Boolean axioms $\{x_j^2 - x_j = 0: j \in [N]\}$ to be used for free.
\end{definition}
\begin{remark}
For the problems we consider in this paper, all of our non-Boolean axioms are monomials, so there is actually no need to use the Boolean axioms.
\end{remark}
\begin{remark}
We allow monomials to contain terms of the form $(1-x_i)$ and allow the Boolean axioms to be used for free in order to avoid spurious lower bounds coming from difficulties in manipulating the Boolean variables rather than handling the non-Boolean axioms. In particular, with these adjustments, when the non-Boolean axioms are monomials, the minimum total coefficient size of a Nullstellensatz proof is upper bounded by the minimum size of a tree-resolution proof.
\end{remark}
Since the Boolean axioms $\{x_j^2 - x_j = 0: j \in [N]\}$ can be used for free, to specify a linear map $D$ from polynomials to $\mathbb{R}$, it is necessary and sufficient to specify the value of $D$ on each input $x \in \{0,1\}^{N}$.
\begin{definition}
Given a function $D: \{0,1\}^{N} \to \mathbb{R}$, we can view $D$ as a linear map from polynomials to $\mathbb{R}$ by taking $D(f) = \sum_{x \in \{0,1\}^{N}}{f(x)D(x)}$
\end{definition}

\section{Total coefficient size lower bound for the pigeonhole principle}
In this section, we prove Theorem \ref{thm:pigeonholelowerbound}, our total coefficient size lower bound on the pigeonhole principle. We start by formally defining the pigeonhole principle.
\begin{definition}[pigeonhole principle ($\mathrm{PHP}_n$)]
	Intuitively, the pigeonhole principle says that if $n$ pigeons are assigned to $n - 1$ holes, then some hole must have more than one pigeon. Formally, for $n \ge 1$, we define $\mathrm{PHP}_n$ to be the statement that the following system of axioms is infeasible:
	\begin{itemize}
		\item For each $i \in [n]$ and $j \in [n-1]$, we have a variable $x_{i, j}$. $x_{i, j} = 1$ represents pigeon $i$ being in hole $j$, and $x_{i, j} = 0$ represents pigeon $i$ not being in hole $j$.
		\item For each $i \in [n]$, we have the axiom $\prod_{j = 1}^{n - 1}{\bar{x}_{i, j}} = 0$
		representing the constraint that each pigeon must be in at least one hole (recall that $\bar{x}_{i,j} = 1 - x_{i,j}$).
		\item For each pair of distinct pigeons $i_1, i_2 \in [n]$ and each hole $j \in [n-1]$, we have the axiom $x_{i_1, j}x_{i_2, j} = 0$
		representing the constraint that pigeons $i_1$ and $i_2$ cannot both be in hole $j$. 
	\end{itemize}
\end{definition}

We prove our lower bound on the total coefficient size complexity of $\text{PHP}_n$ by constructing and analyzing a dual solution $D$. In our dual solution, the only assignments $x$ for which $D(x) \neq 0$ are those where each pigeon goes to exactly one hole (i.e., for each pigeon $i$, exactly one of the $x_{i, j}$ is 1). Note that there are $(n - 1)^n$ such assignments. In the rest of this section, when we refer to assignments or write a summation or expectation over assignments $x$, we refer specifically to these $(n - 1)^n$ assignments.

Recall that the dual constraints are
\[
D(W) = \sum_{\text{assignments } x}{D(x)W(x)} \in [-1,1]
\]
for all weakenings $W$ of an axiom. Note that since $D(x)$ is only nonzero for assignments $x$ where each pigeon goes to exactly one hole, for any weakening $W$ of an axiom  of the form $\prod_{j = 1}^{n - 1}{\bar{x}_{i, j}} = 0$, $D(W) = 0$. Thus, it is sufficient to consider weakenings $W$ of the axioms $x_{i_1, j}x_{i_2, j} = 0$. Further note that if $|D(W)| > 1$ for some weakening $W$ then we can rescale $D$ by dividing by $\max_{W}{|D(W)|}$. Thus, we can rewrite the objective value of the dual program as $\frac{D(1)}{\max_{W}{|D(W)|}}$. Letting $\E$ denote the expectation over a uniform assignment where each pigeon goes to exactly one hole, $\frac{D(1)}{\max_{W}{|D(W)|}} = \frac{\E(D)}{\max_{W}{|\E(DW)|}}$ so it is sufficient to construct $D$ and analyze $\E(D)$ and $\max_{W}{|\E(DW)|}$.

Before constructing and analyzing $D$, we provide some intuition for our construction. The idea is that, if we consider a subset of $n - 1$ pigeons, then $D$ should behave like the indicator function for whether those $n - 1$ pigeons all go to different holes. More concretely, for any polynomial $p$ which does not depend on some pigeon $i$ (i.e. $p$ does not contain $x_{i,j}$ or $\bar{x}_{i,j}$ for any $j \in [n-1]$), 
\[
\E(Dp) = \frac{(n-1)!}{(n-1)^{n-1}}\E(p \mid \text{all pigeons in } [n] \setminus \{i\} \text{ go to different holes})
\]

Given this intuition, we now present our construction. Our dual solution $D$ will be a linear combination of the following functions:

\begin{definition}[functions $J_S$]\label{J}
	Let $S \subsetneq [n]$ be a subset of pigeons of size at most $n - 1$. We define the function $J_S$ that maps assignments to $\{0, 1\}$. For an assignment $x$, $J_S(x) = 1$ if all pigeons in $S$ are in different holes according to $x$, and $J_S(x) = 0$ otherwise. $\qed$
\end{definition}
Note that if $|S| = 0$ or $|S| = 1$, then $J_S$ is the constant function 1. In general, the expectation of $J_S$ over a uniform assignment is $\E(J_S) = \left(\prod_{k = 1}^{|S|} (n - k)\right) / (n - 1)^{|S|}$.\\

\begin{definition}[dual solution $D$]\label{D}
	Our dual solution $D$ is:
	\begin{equation*}
		D = \sum_{S \subsetneq [n]} c_SJ_S,
	\end{equation*}
	where the coefficients $c_S$ are $c_S = \frac{(-1)^{n - 1 - |S|} (n - 1 - |S|)!}{(n - 1)^{n - 1 - |S|}}$.
\end{definition}

We will lower-bound the dual value $\E(D) / \max_W |\E(DW)|$ by computing $\E(D)$ and then upper-bounding $\max_W |\E(DW)|$. In both calculations, we will use the following key property of $D$, which we introduced in our intuition for the construction:

\begin{lemma}\label{dual-intuition}
If $p$ is a polynomial which does not depend on pigeon $i$ (i.e. $p$ does not contain any variables of the form $x_{i,j}$ or $\bar{x}_{i, j}$) then
$\E(Dp) = \E(J_{[n] \setminus \{i\}}p)$.
\end{lemma}
\begin{proof}
Without loss of generality, suppose $p$ does not contain any variables of the form $x_{1,j}$ or $\bar{x}_{1, j}$. Let $T$ be any subset of pigeons that does not contain pigeon 1 and that has size at most $n - 2$. Observe that 
\[
\E({J_{T \cup \{1\}}}p) = \frac{n - 1 - |T|}{n-1}\E({J_{T}}p)
\]
because regardless of the locations of the pigeons in $T$, the probability that pigeon $1$ goes to a different hole is $\frac{n - 1 - |T|}{n-1}$ and $p$ does not depend on the location of pigeon $1$. Since 
\begin{align*}
c_{T \cup \{1\}} &= \frac{(-1)^{n - 2 - |T|} (n - 2 - |T|)!}{(n - 1)^{n - 2 - |T|}} \\
&= -\frac{n-1}{n-1-|T|} \cdot \frac{(-1)^{n - 1 - |T|} (n - 1 - |T|)!}{(n - 1)^{n - 1 - |T|}} = -\frac{n-1}{n-1-|T|}c_{T}
\end{align*}
we have that for all $T \subsetneq \{2, \dots, n\}$,
\[
\E(c_{T \cup \{1\}}{J_{T \cup \{1\}}}p) + \E(c_{T}{J_{T}}p) = 0
\]
Thus, all terms except for $J_{\{2,3,\ldots,n\}}$ cancel. Since $c_{\{2,3,\ldots,n\}} = 1$, we have that $\E(Dp) = \E(J_{\{2,3,\ldots,n\}}p)$, as needed.
\end{proof}

The value of $\E(D)$ follows immediately:

\begin{corollary}\label{exp-d}
	\begin{equation*}
		\E(D) = \frac{(n - 2)!}{(n - 1)^{n - 2}}.
	\end{equation*}
\end{corollary}
\begin{proof} Let $p = 1$. By Lemma \ref{dual-intuition}, $\E(D) = \E(J_{\{2, \dots, n\}}) = (n - 2)!/(n - 1)^{n - 2}$.
\end{proof}

\subsection{Upper bound on $\max_W |\E(DW)|$}

We introduce the following notation:

\begin{definition}[$H_{W, i}$]
Given a weakening $W$, we define a set of holes $H_{W, i} \subseteq [n-1]$ for each pigeon $i \in [n]$ so that $W(x) = 1$ if and only if each pigeon $i \in [n]$ is mapped to one of the holes in $H_{W, i}$. More precisely,
\begin{itemize}
	\item If $W$ contains terms $x_{i, j_1}$ and $x_{i, j_2}$ for distinct holes $j_1, j_2$, then $H_{W, i} = \emptyset$ (i.e. it is impossible that $W(x) = 1$ because pigeon $i$ cannot go to both holes $h$ and $h'$). Similarly, if $W$ contains both $x_{i,j}$ and $\bar{x}_{i,j}$ for some $j$ then $H_{W, i} = \emptyset$ (i.e. it is impossible for pigeon $i$ to both be in hole $j$ and not be in hole $j$).
	\item If $W$ contains exactly one term of the form $x_{i, j}$, then $H_{W, i} = \{j\}$. (i.e., for all $x$ such that $W(x) = 1$, pigeon $i$ goes to hole $j$).
	\item If $W$ contains no terms of the form $x_{i, j}$, then $H_{W, i}$ is the subset of holes $j$ such that $W$ does \textit{not} contain the term $\bar{x}_{i, j}$. (i.e., if $W$ contains the term $\bar{x}_{i, j}$, then for all $x$ such that $W(x) = 1$, pigeon $i$ does not go to hole $j$.)
\end{itemize}
\end{definition}


The key property we will use to bound $\max_W |\E(DW)|$ follows immediately from Lemma \ref{dual-intuition}:

\begin{lemma}\label{exp-dw}
	Let $W$ be a weakening. If there exists some pigeon $i \in [n]$ such that $H_{W, i} = [n-1]$ (i.e., $W$ does not contain any terms of the form $x_{i, j}$ or $\bar{x}_{i, j}$), then $\E(DW) = 0$.
\end{lemma}
\begin{proof}
Without loss of generality, suppose $W$ is a weakening of the axiom $x_{2, 1}x_{3, 1} = 0$ and $H_{W, 1} = [n]$. By Lemma \ref{dual-intuition}, $\E(DW) = \E(J_{\{2, \dots, n\}}W)$. However,  $\E(J_{\{2, \dots, n\}}W) = 0$ because if $W = 1$ then pigeons 2 and 3 must both go to hole 1.
\end{proof}

We make the following definition and then state a corollary of Lemma \ref{exp-dw}.
\begin{definition}[$W^{\mathrm{flip}}_S$]
	Let $W$ be a weakening of the axiom $x_{i_1, j}x_{i_2, j} = 0$ for pigeons $i_1, i_2$ and hole $j$. Let $S \subseteq [n] \setminus \{i_1, i_2\}$. We define $W^{\mathrm{flip}}_S$, which is also a weakening of the axiom $x_{i_1, j}x_{i_2, j} = 0$, as follows.
	\begin{itemize}
		\item For each pigeon $i_3 \in S$, we define $W^{\mathrm{flip}}_S$ so that $H_{W^{\mathrm{flip}}_S, i_3} = [n-1] \setminus H_{W, i_3}$.
		\item For each pigeon $i_3 \notin S$, we define $W^{\mathrm{flip}}_S$ so that $H_{W^{\mathrm{flip}}_S, i_3} = H_{W, i_3}$.
	\end{itemize}
	(Technically, there may be multiple ways to define $W^{\mathrm{flip}}_S$ to satisfy these properties; we can arbitrarily choose any such definition.) $\qed$
\end{definition}

In other words, $W^{\mathrm{flip}}_S$ is obtained from $W$ by flipping the sets of holes that the pigeons in $S$ can go to in order to make the weakening evaluate to 1. Now we state a corollary of Lemma \ref{exp-dw}:\\

\begin{corollary}\label{flip}
	Let $W$ be a weakening of the axiom $x_{i_1, j}x_{i_2, j} = 0$ for pigeons $i_1, i_2$ and hole $j$. Let $S \subseteq [n] \setminus \{i_1, i_2\}$. Then
	\begin{equation*}
		\E\left(DW^{\mathrm{flip}}_S\right) = (-1)^{|S|} \cdot \E(DW).
	\end{equation*}
\end{corollary}
\begin{proof} 
It suffices to show that for $i_3 \in [n] \setminus \{i_1, i_2\}$, we have $\E\left(DW^{\mathrm{flip}}_{\{i_3\}}\right) = -\E(DW)$. Indeed, $W + W^{\mathrm{flip}}_{\{i_3\}}$ is a weakening satisfying $H_{W + W^{\mathrm{flip}}_{\{i_3\}}, i_3} = [n-1]$. Therefore, by Lemma \ref{exp-dw}, $\E\left(D\left(W + W^{\mathrm{flip}}_{\{i_3\}}\right)\right) = 0$.
\end{proof}

Using Corollary \ref{flip}, we can bound $\max_W |\E(DW)|$ using Cauchy-Schwarz. We first show an approach that does not give a strong enough bound. We then show how to modify the approach to achieve a better bound.

\subsubsection{Unsuccessful approach to upper bound $\max_W |\E(DW)|$}\label{unsuccessful}

Consider $\max_W |\E(DW)|$. By Lemma \ref{flip}, it suffices to consider only weakenings $W$ such that, if $W$ is a weakening of the axiom $x_{i_1, j}x_{i_2, j} = 0$, then for all pigeons $i_3 \in [n] \setminus \{i_1, i_2\}$, we have $|H_{W, k}| \leq \lfloor (n - 1) / 2 \rfloor$. For any such $W$, we have
\begin{align*}
	\lVert W \rVert &= \sqrt{E(W^2)}\\
	&\le \sqrt{\left(\frac{1}{n - 1}\right)^2\left(\frac{1}{2}\right)^{n - 2}}\\
	&= (n - 1)^{-1} \cdot 2^{-(n - 2)/2}.
\end{align*}
By Cauchy-Schwarz,
\begin{align*}
	|\E(DW)| &\le \lVert D \rVert \lVert W \rVert\\
	&\le \lVert D \rVert (n - 1)^{-1}2^{-(n - 2)/2}.
\end{align*}
Using the value of $\E(D)$ from Corollary \ref{exp-d}, the dual value $\E(D) / \max_W |\E(DW)|$ is at least:
\[
	\frac{(n - 2)!}{(n - 1)^{n - 2}} \cdot \frac{(n - 1)2^{(n - 2)/2}}{\lVert D \rVert} = \widetilde{\Theta}\left(\left(\frac{e}{\sqrt{2}}\right)^{-n} \cdot \frac{1}{\lVert D \rVert}\right)
\]
by Stirling's formula. Thus, in order to achieve an exponential lower bound on the dual value, we would need $1 / \lVert D \rVert \ge \Omega(c^n)$ for some $c > e/\sqrt{2}$. However, this requirement is too strong, as we will show that $1 / \lVert D \rVert = \widetilde{\Theta}\left(\left(\sqrt{e}\right)^n\right)$. Directly applying Cauchy-Schwarz results in too loose of a bound on $\max_W |\E(DW)|$, so we now modify our approach.

\subsubsection{Successful approach to upper bound $\max_W |\E(DW)|$}

\begin{definition}[$W^{\{-1, 0, 1\}}$]
	Let $W$ be a weakening of the axiom $x_{i_1, j}x_{i_2, j} = 0$ for pigeons $i_1, i_2$ and hole $j$. We define the function $W^{\{-1, 0, 1\}}$ that maps assignments to $\{-1, 0, 1\}$. For an assignment $x$,
	\begin{itemize}
		\item If pigeons $i_1$ and $i_2$ do not both go to hole $j$, then $W^{\{-1, 0, 1\}}(x) = 0$.
		\item Otherwise, let $V(x) = |\{i_3 \in [n] \setminus \{i_1, i_2\} : \text{pigeon } i_3 \text{ does not go to } H_{W, i_3}\}|$. Then $W^{\{-1, 0, 1\}}(x) = (-1)^{V(x)}$.
	\end{itemize}
\end{definition}

Note that $W^{\{-1, 0, 1\}}$ is a linear combination of the $W^{\mathrm{flip}}_S$:\\

\begin{lemma}\label{exp-dw-plus-minus}
	Let $W$ be a weakening of the axiom $x_{i_1, j}x_{i_2, j} = 0$ for pigeons $i_1, i_2$ and hole $j$. We have:
	\begin{equation*}
		W^{\{-1, 0, 1\}} = \sum_{S \subseteq [n] \setminus \{i_1, i_2\}} (-1)^{|S|} \cdot W^{\mathrm{flip}}_S.
	\end{equation*}
	It follows that:
	\begin{equation*}
		\E\left(DW^{\{-1, 0, 1\}}\right) = 2^{n - 2} \cdot \E(DW).
	\end{equation*}
\end{lemma}
\begin{proof}To prove the first equation, consider any assignment $x$. If pigeons $i_1$ and $i_2$ do not both go to hole $j$, then both $W^{\{-1, 0, 1\}}$ and all the $W^{\mathrm{flip}}_S$ evaluate to 0 on $x$. Otherwise, exactly one of the $W^{\mathrm{flip}}_S(x)$ equals 1, and for this choice of $S$, we have $W^{\{-1, 0, 1\}}(x) = (-1)^{|S|}$.

The second equation follows because:
\begin{align*}
	\E\left(DW^{\{-1, 0, 1\}}\right) &= \sum_{S \subseteq [n] \setminus \{i_1, i_2\}} (-1)^{|S|} \cdot \E\left(DW^{\mathrm{flip}}_S\right)\\
	&= \sum_{S \subseteq [n] \setminus \{i_1, i_2\}} (-1)^{|S|}(-1)^{|S|} \cdot \E(DW) \tag{Corollary \ref{flip}}\\
	&= 2^{n - 2} \cdot \E(DW).
\end{align*}
\end{proof}

Using Lemma \ref{exp-dw-plus-minus}, we now improve on the approach to upper-bound $\max_W |\E(DW)|$ from section \ref{unsuccessful}:

\begin{lemma}\label{exp-DW-successful}
	The dual value $\E(D) / \max_W |\E(DW)|$ is at least $\frac{(n - 2)!}{(n - 1)^{n - 2}} \cdot \frac{(n - 1)2^{n - 2}}{\lVert D \rVert}$
\end{lemma}

\begin{proof} 
For any $W$, we have:
\begin{align*}
	\E(DW) &= 2^{-(n - 2)} \cdot \E\left(DW^{\{-1, 0, 1\}}\right) \tag{Lemma \ref{exp-dw-plus-minus}}\\
	&\le 2^{-(n - 2)} \cdot \lVert D \rVert \lVert W^{\{-1, 0, 1\}} \rVert \tag{Cauchy-Schwarz}\\
	&= 2^{-(n - 2)} \cdot \lVert D \rVert \sqrt{\E\left(\left(W^{\{-1, 0, 1\}}\right)^2\right)}\\
	&= (n - 1)^{-1}2^{-(n - 2)} \cdot \lVert D \rVert.
\end{align*}
Using the value of $\E(D)$ from Corollary \ref{exp-d}, the dual value $\E(D) / \max_W |\E(DW)|$ is at least $\frac{(n - 2)!}{(n - 1)^{n - 2}} \cdot \frac{(n - 1)2^{n - 2}}{\lVert D \rVert}$.
\end{proof}

It only remains to compute $\lVert D \rVert$:\\

\begin{lemma}\label{norm-D}
\[
{\lVert D \rVert}^2 = \frac{(n - 2)!}{(n - 1)^{n - 2}} \cdot n! \cdot \sum_{c = 0}^{n - 1} \frac{(-1)^{n - 1 - c}}{n - c} \cdot \frac{1}{(n - 1)^{n - 1 - c}c!}
\]
\end{lemma}

\begin{proof}
Recall the definition of $D$ (Definition \ref{D}):
\begin{align*}
	D &= \sum_{S \subsetneq [n]} c_SJ_S,\\
	c_S &= \frac{(-1)^{n - 1 - |S|} (n - 1 - |S|)!}{(n - 1)^{n - 1 - |S|}}.
\end{align*}
We compute $\lVert D \rVert^2 = \E(D^2)$ as follows.
\begin{equation*}
	\E(D^2) = \sum_{S \subsetneq [n]} \sum_{T \subsetneq [n]} c_Sc_T \cdot \E(J_SJ_T).
\end{equation*}
Given $S, T \subsetneq [n]$, we have:
\begin{align*}
	\E(J_SJ_T) &= \E(J_S)\E(J_T \mid J_S = 1)\\
	&= \left(\left(\prod_{i = 1}^{|S|} (n - i)!\right) / (n - 1)^{|S|}\right)\left(\left(\prod_{j = |S \cap T| + 1}^{|T|} (n - j)!\right) / (n - 1)^{|T \setminus S|}\right)
\end{align*}
Therefore,
\begin{align*}
	c_Sc_T \cdot \E(J_SJ_T) &= \left(c_S\left(\prod_{i = 1}^{|S|} (n - i)!\right) / (n - 1)^{|S|}\right)\left(c_T\left(\prod_{j = |S \cap T| + 1}^{|T|} (n - j)!\right) / (n - 1)^{|T \setminus S|}\right).
\end{align*}
Note that the product of $(-1)^{n - 1 - |S|}$ (from the $c_S$) and $(-1)^{n - 1 - |T|}$ (from the $c_T$) equals $(-1)^{|S| - |T|}$, so the above equation becomes:
\begin{align*}
	c_Sc_T \cdot \E(J_SJ_T) &= (-1)^{|S| - |T|} \left(\frac{(n - 2)!}{(n - 1)^{n - 2}}\right)\left(\frac{(n - 1 - |S \cap T|)!}{(n - 1)^{n - 1 - |S \cap T|}}\right).
\end{align*}
Now, we rearrange the sum for $\E(D^2)$ in the following way:
\begin{align*}
	\E(D^2) &= \sum_{S \subsetneq [n]} \sum_{T \subsetneq [n]} c_Sc_T \cdot \E(J_SJ_T)\\
	&= \frac{(n - 2)!}{(n - 1)^{n - 2}} \sum_{c = 0}^{n - 1} \frac{(n - 1 - c)!}{(n - 1)^{n - 1 - c}} \sum_{\substack{S, T \subsetneq [n],\\|S \cap T| = c}} (-1)^{|S| - |T|}.
\end{align*}

To evaluate this expression, fix $c \le n - 1$ and consider the inner sum. Consider the collection of tuples $\{(S, T) \mid S, T \subsetneq [n], |S \cap T| = c\}$. We can pair up (most of) these tuples in the following way. For each $S$, let $m_S$ denote the minimum element in $[n]$ that is not in $S$ (note that $m_S$ is well defined because $S$ cannot be $[n]$). We pair up the tuple $(S, T)$ with the tuple $(S, T \triangle \{m_S\})$, where $\triangle$ denotes symmetric difference. The only tuples $(S, T)$ that cannot be paired up in this way are those where $|S| = c$ and $T = [n] \setminus \{m_S\}$, because $T$ cannot be $[n]$. There are $\binom{n}{c}$ unpaired tuples $(S, T)$, and for each of these tuples, we have $(-1)^{|S| - |T|} = (-1)^{n - 1 - c}$. On the other hand, each pair $(S, T), (S, T \triangle \{m_S\})$ contributes 0 to the inner sum. Therefore, the inner sum equals $(-1)^{n - 1 - c}\binom{n}{c}$, and we have:
\begin{align*}
	\E(D^2) &= \frac{(n - 2)!}{(n - 1)^{n - 2}} \sum_{c = 0}^{n - 1} \frac{(-1)^{n - 1 - c}(n - 1 - c)!}{(n - 1)^{n - 1 - c}}\binom{n}{c}\\
	&= \frac{(n - 2)!}{(n - 1)^{n - 2}} \sum_{c = 0}^{n - 1} \frac{(-1)^{n - 1 - c}(n - 1 - c)!}{(n - 1)^{n - 1 - c}} \cdot \frac{n!}{c!(n - c)!}\\
	&= \frac{(n - 2)!}{(n - 1)^{n - 2}} \cdot n! \cdot \sum_{c = 0}^{n - 1} \frac{(-1)^{n - 1 - c}}{n - c} \cdot \frac{1}{(n - 1)^{n - 1 - c}c!}.
\end{align*}
\end{proof}
\begin{corollary}\label{cor:roughnormbound}
$\E(D^2) \leq \frac{n!}{(n-1)^{n-1}}$
\end{corollary}
\begin{proof}
Observe that the sum 
\[
\sum_{c = 0}^{n - 1} \frac{(-1)^{n - 1 - c}}{n - c} \cdot \frac{1}{(n - 1)^{n - 1 - c}c!}
\] 
is an alternating series where the magnitudes of the terms decrease as $c$ decreases. The two largest magnitude terms are $1/(n - 1)!$ and $-(1/2) \cdot 1/(n - 1)!$. Therefore, the sum is at most $\frac{1}{(n - 1)!}$, and we conclude that
\[
\E(D^2) \leq \frac{(n - 2)!}{(n - 1)^{n - 2}} \cdot \frac{n!}{(n-1)!} = \frac{n!}{(n-1)^{n-1}} 
\]
as needed.
\end{proof}
We can now complete the proof of Theorem \ref{thm:pigeonholelowerbound} \begin{proof}[Proof of Theorem \ref{thm:pigeonholelowerbound}]
By Lemma \ref{exp-DW-successful}, any Nullstellensatz proof for $\text{PHP}_n$ has total coefficient size at least $\frac{(n - 2)!}{(n - 1)^{n - 2}} \cdot \frac{(n - 1)2^{n - 2}}{\lVert D \rVert}$. By Corollary \ref{cor:roughnormbound}, $\lVert D \rVert \leq \sqrt{\frac{n!}{(n-1)^{n-1}}}$. Combining these results,
any Nullstellensatz proof for $\text{PHP}_n$ has total coefficient size at least 
\begin{align*}
    \frac{(n - 2)!}{(n - 1)^{n - 2}} \cdot \frac{(n - 1)2^{n - 2}}{\sqrt{\frac{n!}{(n-1)^{n-1}}}} &= \frac{2^{n-2}}{\sqrt{n}} \cdot \frac{\sqrt{(n-1)!}}{(n-1)^{\frac{n}{2} - \frac{3}{2}}} \\
    &= \frac{2^{n-2}(n-1)}{\sqrt{n}}\sqrt{\frac{(n-1)!}{(n-1)^{n-1}}}
\end{align*}
Using Stirling's approximation that $n!$ is approximately $\sqrt{2{\pi}n}\left(\frac{n}{e}\right)^n$, $\sqrt{\frac{(n-1)!}{(n-1)^{n-1}}}$ is approximately $\sqrt[4]{2{\pi}(n-1)}\left(\frac{1}{\sqrt{e}}\right)^{n-1}$ so this expression is $\Omega\left(n^{\frac{3}{4}}\left(\frac{2}{\sqrt{e}}\right)^{n}\right)$, as needed.
\end{proof}


\subsection{Experimental Results for $\text{PHP}_n$}


For small $n$, we computed the optimal dual values shown below. The first column of values is the optimal dual value for $n = 3, 4$. The second column of values is the optimal dual value for $n = 3, 4, 5, 6$ under the restriction that the only nonzero assignments are those where each pigeon goes to exactly one hole.

\begin{center}
\begin{tabular}{ |c|c|c| } 
 \hline
 $n$ & dual value & dual value, each pigeon goes to exactly one hole \\
 \hline
 3 & 11 & 6 \\ 
 4 & $41.4\overline{69}$ & 27 \\
 5 & - & 100 \\
 6 & - & 293.75 \\
 \hline
\end{tabular}
\end{center}

For comparison, the table below shows the value we computed for our dual solution and the lower bound of $\frac{2^{n-2}(n-1)}{\sqrt{n}}\sqrt{\frac{(n-1)!}{(n-1)^{n-1}}}$ that we showed in the proof of Theorem \ref{thm:pigeonholelowerbound}. (Values are rounded to 3 decimals.)

\begin{center}
\begin{tabular}{ |c|c|c| } 
 \hline
 $n$ & value of $D$ & proven lower bound on value of $D$ \\
 \hline
 3 & 4 & 1.633 \\ 
 4 & 18 & 2.828 \\
 5 & 64 & 4.382 \\
 6 & 210.674 & 6.4 \\
 \hline
\end{tabular}
\end{center}

It is possible that our lower bound on the value of $D$ can be improved. The following experimental evidence suggests that the dual value $\E(D) / \max_W |\E(DW)|$ of $D$ may actually be $\widetilde{\Theta}(2^n)$. For $n = 3, 4, 5, 6$, we found that the weakenings $W$ that maximize $|\E(DW)|$ are of the following form, up to symmetry. (By symmetry, we mean that we can permute pigeons/holes without changing $|\E(DW)|$, and we can flip sets of holes as in Lemma \ref{flip} without changing $|\E(DW)|$.)
\begin{itemize}
    \item For odd $n$ ($n = 3, 5$): $W$ is the weakening of the axiom $x_{1, 1}x_{2, 1} = 0$ where, for $i = 3, \dots, n$, we have $H_{W, i} = \{2, \dots, (n + 1)/2\}$.
    \item For even $n$ ($n = 4, 6$): $W$ is the following weakening of the axiom $x_{1, 1}x_{2, 1} = 0$. For $i = 3, \dots, n/2 + 1$, we have $H_{W, i} = \{2, \dots, n/2\}$. For $i = n/2 + 2, \dots, n$, we have $H_{W, i} = \{n/2 + 1, \dots, n - 1\}$.
\end{itemize}
If this pattern continues to hold for larger $n$, then experimentally it seems that\\$\E(D) / \max_W |\E(DW)|$ is $\widetilde{\Theta}(2^n)$, although we do not have a proof of this.

\section{Total coefficient size upper bound for the ordering principle}
In this section, we construct an explicit Nullstellensatz proof of infeasibility for the ordering principle $\text{ORD}_n$ with size and total coefficient size $2^n - n$. We start by formally defining the ordering principle.

\begin{definition}[ordering principle ($\mathrm{ORD}_n$)]
	Intuitively, the ordering principle says that any well-ordering on $n$ elements must have a minimum element. Formally, for $n \ge 1$, we define $\mathrm{ORD}_n$ to be the statement that the following system of axioms is infeasible:
	\begin{itemize}
		\item We have a variable $x_{i, j}$ for each pair $i, j \in [n]$ with $i < j$. $x_{i, j} = 1$ represents element $i$ being less than element $j$ in the well-ordering, and $x_{i, j} = 0$ represents element $i$ being more than element $j$ in the well-ordering.
		
        We write $x_{j,i}$ as shorthand for $1 - x_{i,j}$ (i.e. we take $x_{j,i} = \bar{x}_{i,j} = 1 - x_{i,j}$).
		\item For each $i \in [n]$, we have the axiom
		$\prod_{j \in [n] \setminus \{i\}}{x_{i,j}} = 0$
		which represents the constraint that element $i$ is not a minimum element. We call these axioms non-minimality axioms.
		\item For each triple $i,j,k \in [n]$ where $i < j < k$, we have the two axioms $x_{i,j}x_{j,k}x_{k,i} = 0$ and $x_{k,j}x_{j,i}x_{i,k} = 0$ which represent the constraints that elements $i, j, k$ satisfy transitivity. We call these axioms transitivity axioms.
	\end{itemize}
\end{definition}

%

In our Nullstellensatz proof, for each weakening $W$ of an axiom, its coefficient $c_W$ will either be $1$ or $0$. Non-minimality axioms will appear with coefficient $1$ and the only weakenings of transitivity axioms which appear have a special form which we describe below.

\begin{definition}[nice transitivity weakening]
	Let $W$ be a weakening of the axiom $x_{i,j}x_{j,k}x_{k,i}$ or the axiom $x_{k,j}x_{j,i}x_{i,k}$ for some $i < j < k$. Let $G(W)$ be the following directed graph. The vertices of $G(W)$ are $[n]$. For distinct $i', j' \in [n]$, $G(W)$ has an edge from $i'$ to $j'$ if $W$ contains the term $x_{i', j'}$. We say that $W$ is a \textit{nice transitivity weakening} if $G(W)$ has exactly $n$ edges and all vertices are reachable from vertex $i$.
\end{definition}

In other words, if $W$ is a weakening of the axiom $x_{i,j}x_{j,k}x_{k,i}$ or the axiom $x_{k,j}x_{j,i}x_{i,k}$ then $G(W)$ contains a 3-cycle on vertices $\{i, j, k\}$. $W$ is a nice transitivity weakening if and only if contracting this 3-cycle results in a (directed) spanning tree rooted at the contracted vertex. Note that if $W$ is a nice transitivity weakening and $x$ is an assignment with a minimum element then $W(x) = 0$.\\

\begin{theorem}\label{ordering-primal}
	There is a Nullstellensatz proof of infeasibility for $\text{ORD}_n$ satisfying:
	\begin{enumerate}
		\item The total coefficient size is $2^n - n$.
		\item Each $c_W$ is either 0 or 1.
		\item If $A$ is a non-minimality axiom, then $c_A = 1$ and $c_W = 0$ for all other weakenings of $A$.
		\item If $W$ is a transitivity weakening but not a nice transitivity weakening then $c_W = 0$.
	\end{enumerate}
\end{theorem}

\textbf{Proof.} We prove Theorem \ref{ordering-primal} by induction on $n$. When $n = 3$, the desired Nullstellensatz proof sets $c_A = 1$ for each axiom $A$. It can be verified that $\sum_W c_WW$ evaluates to 1 on each assignment, and that this Nullstellensatz proof satisfies the properties of Theorem \ref{ordering-primal}.

Now suppose we have a Nullstellensatz proof for $\text{ORD}_n$ satisfying Theorem \ref{ordering-primal}, and let $S_n$ denote the set of transitivity weakenings $W$ for which $c_W = 1$. The idea to obtain a Nullstellensatz proof for $\text{ORD}_{n + 1}$ is to use two ``copies'' of $S_n$, the first copy on elements $\{1, \dots, n\}$ and the second copy on elements $\{2, \dots, n + 1\}$. Specifically, we construct the Nullstellensatz proof for $\text{ORD}_{n + 1}$ by setting the following $c_W$ to 1 and all other $c_W$ to 0.
\begin{enumerate}
	\item For each non-minimality axiom $A$ in $\text{ORD}_{n + 1}$, we set $c_A = 1$.
	\item For each $W \in S_n$, we define the transitivity weakening $W'$ on $n + 1$ elements by $W' = W \cdot x_{1, n + 1}$ and set $c_{W'} = 1$.
	\item For each $W \in S_n$, first we define the transitivity weakening $W''$ on $n + 1$ elements by replacing each variable $x_{i, j}$ that appears in $W$ by $x_{i + 1, j + 1}$. (e.g., if $W = x_{1, 2}x_{2, 3}x_{3,1}$, then $W'' = x_{2, 3}x_{3, 4}x_{4,2}$.) Then, we define $W''' = W''x_{n + 1,1}$ and set $c_{W'''} = 1$.
	\item For each $i \in \{2, \dots, n\}$, for each of the 2 transitivity axioms $A$ on $(1, i, n + 1)$, we set $c_W = 1$ for the following weakening $W$ of $A$:
	\begin{equation*}
		W = A\left(\prod_{j \in [n] \setminus \{i\}}{x_{i, j}}\right).
	\end{equation*}
	In other words, $W(x) = 1$ if and only if $A(x) = 1$ and $i$ is the minimum element among the elements $[n+1] \setminus \{1, n + 1\}$.
\end{enumerate}

The desired properties 1 through 4 in Theorem \ref{ordering-primal} can be verified by induction. It remains to show that for each assignment $x$, there is exactly one nonzero $c_W$ for which $W(x) = 1$. If $x$ has a minimum element $i \in [n+1]$, then the only nonzero $c_W$ for which $W(x) = 1$ is the non-minimality axiom for $i$. Now suppose that $x$ does not have a minimum element. Consider two cases: either $x_{1, n + 1} = 1$, or $x_{n + 1,1} = 1$. Suppose $x_{1, n + 1} = 1$. Consider the two subcases:
\begin{enumerate}
    \item Suppose that, if we ignore element $n + 1$, then there is still no minimum element among the elements $\{1, \dots, n\}$. Then there is exactly one weakening $W$ in point 2 of the construction for which $W(x) = 1$, by induction.
    \item Otherwise, for some $i \in \{2, \dots, n\}$, we have that $i$ is a minimum element among $\{1, \dots, n\}$ and $x_{n + 1,i} = 1$. Then there is exactly one weakening $W$ in point 4 of the construction for which $W(x) = 1$ (namely the weakening $W$ of the axiom $A = x_{i,1}x_{1, n + 1}x_{n+1,i}$).
\end{enumerate}
The case $x_{n + 1,1} = 1$ is handled similarly by considering whether there is a minimum element among $\{2, \dots, n + 1\}$. Assignments that do have a minimum element among $\{2, \dots, n + 1\}$ are handled by point 3 of the construction, and assignments that do not are handled by point 4 of the construction. $\qed$
\subsection{Restriction to instances with no minimial element}
We now observe that for the ordering principle, we can restrict our attention to instances which have no minimum element.
\begin{lemma}
Suppose we have coefficients $c_W$ satisfying $\sum_W c_{W}W(x) = 1$ for all assignments $x$ that have no minimum element (but it is possible that $\sum_W c_{W}W(x) \neq 1$ on assignments $x$ that do have a minimum element). Then there exist coefficients $c'_{W}$ such that $\sum_W c'_{W}W = 1$ (i.e., the coefficients $c_W'$ are a valid primal solution) with
\begin{equation*}
    \sum_{W}{|c'_W|} \leq (n + 1)\left(\sum_{W}{ |c_W|}\right) + n.
\end{equation*}
\end{lemma}
This lemma says that, to prove upper or lower bounds for $\text{ORD}_n$ by constructing primal or dual solutions, it suffices to consider only assignments $x$ that have no minimum element, up to a factor of $O(n)$ in the solution value.
\begin{proof}
Let $C$ denote the function on weakenings that maps $W$ to $c_W$. For $i \in [n]$, we will define the function $C_i$ on weakenings satisfying the properties:
\begin{itemize}
    \item If $x$ is an assignment where $i$ is a minimum element, then $\sum_{W}{C_i(W)W(x)} = \sum_{W}{C(W)W(x)}$.
    \item Otherwise, $\sum_{W}{C_i(W)W(x)} = 0$.
\end{itemize}
Let $A_i = \prod_{j\in [n] \setminus \{i\}}{x_{i, j}}$ be the non-minimality axiom for $i$. Intuitively, we want to define $C_i$ as follows: For all $W$, $C_i(A_iW) = C(W)$. (If $W$ is a weakening that is not a weakening of $A_i$, then $C_i(W) = 0$.) The only technicality is that multiple weakenings $W$ may become the same when multiplied by $A_i$, so we actually define $C_i(A_iW) = \sum_{W': A_iW' = A_iW} C(W')$.

Finally, we use the functions $C_i$ to define the function $C'$:
\begin{equation*}
    C' = C - \left(\sum_{i = 1}^n C_i\right) + \left(\sum_{i = 1}^n A_i\right).
\end{equation*}
By taking $c'_W = C'(W)$, the $c'_W$ are a valid primal solution with the desired bound on the total coefficient size.
\end{proof}
\subsection{Experimental results}
For small values of $n$, we have computed both the minimum total coefficient size of a Nullstellensatz proof of the ordering principle and the value of the linear program where we restrict our attention to instances $x$ which have no minimum element.

We found that for $n = 3,4,5$, the minimum total coefficient size of a Nullstellensatz proof of the ordering principle is $2^n - n$ so the primal solution given by Theorem \ref{ordering-primal} is optimal. However, for $n = 6$ this solution is not optimal as the minimum total coefficient size is $52$ rather than $2^6 - 6 = 58$.

If we restrict our attention to instances $x$ which have no minimum element then for $n = 3,4,5,6$, the value of the resulting linear program is equal to $2\binom{n}{3}$, which is the number of transitivity axioms. However, this is no longer true for $n = 7$, though we did not compute the exact value.

\section{Analyzing Total Coefficient Size for Stronger Proof Systems}
In this section, we consider the total coefficient size for two stronger proof systems, sum of squares proofs and a proof system which is between Nullstellensatz and sum of squares proofs which we call resolution-like proofs.
\begin{definition}
Given a system of axioms $\{p_i = 0: i \in [m]\}$, we define a resolution-like proof of infeasibility to be an equality of the form
\[
-1 = \sum_{i=1}^{m}{{p_i}{q_i}} + \sum_{j}{{c_j}g_j}
\]
where each $g_j$ is a monomial and each coefficient $c_j$ is non-negative. We define the total coefficient size of such a proof to be $\sum_{i=1}^{m}{T(q_i)} + \sum_{j}{c_j}$.
\end{definition}
We call this proof system resolution-like because it captures the resolution-like calculus introduced for Max-SAT by Mar\'{i}a Luisa Bonet, Jordi Levy, and Felip Many\`{a}
\cite{BONET2007606}. The idea is that if we have deduced that $x{r_1} \leq 0$ and $\bar{x}{r_2} \leq 0$ for some variable $x$ and monomials $r_1$ and $r_2$ then we can deduce that ${r_1}{r_2} \leq 0$ as follows:
\[
{r_1}{r_2} = x{r_1} - (1 - r_2)x{r_1} + \bar{x}{r_2} - (1 - r_1)\bar{x}{r_2}
\]
where we decompose $(1 - r_1)$ and $(1-r_2)$ into monomials using the observation that $1 - \prod_{i=1}^{k}{x_i} = \sum_{j = 1}^{k}{(1 - x_j)\left(\prod_{i=1}^{j-1}{x_i}\right)}$.

The minimum total coefficient size of a resolution-like proof can be found using the following linear program.
\begin{enumerate}
    \item[] Primal: Minimize $\sum_{i=1}^{m}{T(q_i)} + \sum_{j}{c_j}$ subject to $\sum_{i=1}^{m}{{p_i}{q_i}} + \sum_{j}{{c_j}g_j} = -1$
    \item[] Dual: Maximize $D(1)$ subject to the constraints that
    \begin{enumerate}
        \item[1.] $D$ is a linear map from polynomials to $\mathbb{R}$.
        \item[2.] For each $i \in [m]$ and each monomial $r$, $|D(rp_i)| \leq 1$.
        \item[3.] For each monomial $r$, $D(r) \geq -1$.
    \end{enumerate}
\end{enumerate}
\begin{definition}
Given a system of axioms $\{p_i = 0: i \in [m]\}$, a Positivstellensatz/sum of squares proof of infeasibility is an equality of the form
\[
-1 = \sum_{i=1}^{m}{{p_i}{q_i}} + \sum_{j}{g_j^2}
\]
We define the total coefficient size of a Positivstellensatz/sum of squares proof to be $\sum_{i=1}^{m}{T(q_i)} + \sum_{j}{T(g_j)^2}$
\end{definition}

\begin{enumerate}
    \item[] Primal: Minimize $\sum_{i=1}^{m}{T(q_i)} + \sum_{j}{T(g_j)^2}$ subject to the constraint that $-1 = \sum_{i=1}^{m}{{p_i}{q_i}} + \sum_{j}{g_j^2}$.
    \item[] Dual: Maximize $D(1)$ subject to the constraints that
    \begin{enumerate}
        \item[1.] $D$ is a linear map from polynomials to $\mathbb{R}$.
        \item[2.] For each $i \in [m]$ and each monomial $r$, $|D(rp_i)| \leq 1$.
        \item[3.] For each polynomial $g_j$, $D((g_j)^2) \geq -T(g_j)^2$,
    \end{enumerate}
\end{enumerate}
\subsection{Failure of the dual certificate for resolution-like proofs}
In this subsection, we observe that our dual certificate does not give a lower bound on the total coefficient size for resolution-like proofs of the pigeonhole principle because it has a large negative value on some monomials.
\begin{theorem}
The value of the dual certificate on the polynomial $\prod_{i=1}^{n}{\bar{x}_{i1}}$ is \\ $-\frac{(n-2)!}{(n-1)^{n-1}}\left(1 - \frac{(-1)^{n - 1}}{(n-1)^{n-2}}\right)$
\end{theorem}
\begin{proof}
To show this, we make the following observations.
\begin{enumerate}
    \item The value of the dual certificate on the polynomial $\prod_{i=2}^{n}{\bar{x}_{i1}}$ is $0$.
    \item The value of the dual certificate on the polynomial $x_{11}\prod_{i=3}^{n}{\bar{x}_{i1}}$ is $\frac{(n-2)!}{(n-1)^{n-1}}$
    \item The value of the dual certificate on the polynomial $x_{11}x_{21}\prod_{i=3}^{n}{\bar{x}_{i1}}$ is $\frac{(-1)^{n-2}(n-2)!}{(n-1)^{2n-3}}$.
\end{enumerate}
For the first observation, observe that since the first pigeon is unrestricted, every term of the dual certificate cancels except $J_{\{2,3,\ldots,n\}}$ which is $0$ as none of these pigeons can go to hole $1$. For the second observation, observe that since the second pigeon is unrestriced, every term of the dual certificate cancels except $J_{\{1,3,4,\ldots,n\}}$ which gives value $\frac{\E(D)}{n-1} = \frac{(n-2)!}{(n-1)^{n-1}}$. For the third observation, observe that by Lemma \ref{flip}, the value of the dual certificate on the polynomial $x_{11}x_{21}\prod_{i=3}^{n}{\bar{x}_{i1}}$ is $(-1)^{n-2}$ times the value of the dual certificate on the polynomial $\prod_{i=1}^{n}{x_{i1}}$ which is 
\[
\frac{1}{(n-1)^{n}}\left(-\frac{(n-1)!}{(n-1)^{n-1}} + n\frac{(n-2)!}{(n-1)^{n-2}}\right) = \frac{(n-2)!}{(n-1)^{2n-3}}
\]
Putting these observations together, the value of the dual certificate for the polynomial 
\[
\prod_{i=1}^{n}{\bar{x}_{i1}} = \prod_{i=2}^{n}{\bar{x}_{i1}} - x_{11}\prod_{i=3}^{n}{\bar{x}_{i1}} + x_{11}x_{21}\prod_{i=3}^{n}{\bar{x}_{i1}}
\]
is $-\frac{(n-2)!}{(n-1)^{n-1}} + \frac{(-1)^{n - 2}(n-2)!}{(n-1)^{2n-3}} = -\frac{(n-2)!}{(n-1)^{n-1}}\left(1 - \frac{(-1)^{n - 1}}{(n-1)^{n-2}}\right)$.
\end{proof}
\subsection{Small total coefficient size sum of squares proof of the ordering principle}
In this subsection, we show that the small size resolution proof of the ordering principle \cite{staalmarck1996short}, which seems to be dynamic in nature, can actually be mimicked by a sum of squares proof. Thus, while sum of squares requires degree $\tilde{\Theta}(\sqrt{n})$ to refute the negation of the ordering principle \cite{potechin:LIPIcs:2020:12590}, there is a sum of squares proof which has polynomial size and total coefficient size. To make our proof easier to express, we define the following monomials
\begin{definition} \ 
\begin{enumerate}
\item Whenever $1 \leq j \leq m \leq n$, let 
$F_{jm} = \prod_{i \in [m] \setminus \{j\}}{x_{ji}}$ 
be the monomial which is $1$ if $x_j$ is the first element in $x_1,\dots,x_m$ and $0$ otherwise.
\item For all $m \in [n-1]$ and all distinct $j,k \in [m]$, we define $T_{jmk}$ to be 
the monomial 
\[
T_{jmk} = F_{jm}x_{(m+1)j}x_{k(m+1)}\prod_{i \in [k-1] \setminus \{j\}}{x_{(m+1)i}}
\]
Note that $T_{jmk}$ is a multiple of $x_{(m+1)j}x_{jk}x_{k(m+1)}$ so it is a weakening of a transitivity axiom.
\end{enumerate}
\end{definition}
With these definitions, we can now express our proof.
\begin{theorem}
The following equality (modulo the axioms that $x_{ij}^2 = x_{ij}$ and $x_{ij}x_{ji} = 0$ for all distinct $i,j \in [n]$) gives an SoS proof that the total ordering axioms are infeasible.
\[
-1 = \sum_{m=1}^{n-1}{\left(\left(F_{(m+1)(m+1)} - \sum_{j=1}^{m}{F_{jm}F_{(m+1)(m+1)}}\right)^2 - \sum_{j=1}^{m}{\sum_{k \in [m] \setminus \{j\}}{T_{jmk}}}\right)} 
- \sum_{j=1}^{n}{F_{jn}}
\]
\end{theorem}
\begin{proof}
Our key building block is the following lemma.
\begin{lemma}\label{lem:buildingblock}
For all $m \in [1,n-1]$ and all $j \in [1,m]$, 
\[
F_{jm} = F_{j(m+1)} + \sum_{k \in [m] \setminus \{j\}}{T_{jmk}} + 
F_{jm}F_{(m+1)(m+1)}
\]
\end{lemma}
\begin{proof}
First note that 
\[
F_{jm} = F_{jm}x_{j(m+1)} + F_{jm} x_{(m+1)j} = F_{j(m+1)} + F_{jm}x_{(m+1)j}
\]
We now use the following proposition
\begin{proposition}
\[
F_{jm}x_{(m+1)j} = \sum_{k \in [m] \setminus \{j\}}{F_{jm}x_{(m+1)j}x_{k(m+1)}\prod_{i \in [k-1] \setminus \{j\}}{x_{(m+1)i}}} + F_{jm}F_{(m+1)(m+1)}x_{(m+1)j}
\]
\end{proposition}
\begin{proof}
If $x_{k'(m+1)} = 0$ for all $k \in [m]$ then $x_{k(m+1)}\prod_{i \in [k-1] \setminus \{j\}}{x_{(m+1)i}} = 0$ for all $k \in [m]$ and $F_{(m+1)(m+1)} = 1$. Otherwise, let $k'$ be the first index in $[m]$ such that  $x_{k'(m+1)} = 1$. Now observe that $F_{(m+1)(m+1)} = 0$ and 
$x_{k(m+1)}\prod_{i \in [k-1] \setminus \{j\}}{x_{(m+1)i}} = 1$ if $k = k'$ and is $0$ if $k \neq k'$.
\end{proof}
Finally, we observe that $F_{jm}F_{(m+1)(m+1)}x_{(m+1)j} = F_{jm}F_{(m+1)(m+1)}$ as $x_{(m+1)j}$ is contained in $F_{(m+1)(m+1)}$. Putting everything together, 
\[
F_{jm} = F_{j(m+1)} + \sum_{k \in [m] \setminus \{j\}}{T_{jmk}} + 
F_{jm}F_{(m+1)(m+1)}
\]
\end{proof}

We now verify that our proof, which we restate here for convenience, is indeed an equality modulo the axioms that $x_{ij}^2 = x_{ij}$ and $x_{ij}x_{ji} = 0$ for all distinct $i,j \in [n]$.
\[
-1 = \sum_{m=1}^{n-1}{\left(\left(F_{(m+1)(m+1)} - \sum_{j=1}^{m}{F_{jm}F_{(m+1)(m+1)}}\right)^2 - \sum_{j=1}^{m}{\sum_{k \in [m] \setminus \{j\}}{T_{jmk}}}\right)} 
- \sum_{j=1}^{n}{F_{jn}}
\]
Observe that $F_{(m+1)(m+1)}^2 = F_{(m+1)(m+1)}$, $F_{jm}^2 = F_{(jm)}$, and for all distinct $j,j' \in [m]$, $F_{jm}F_{j'm} = 0$.
Thus,
\[
\left(F_{(m+1)(m+1)} - \sum_{j=1}^{m}{F_{jm}F_{(m+1)(m+1)}}\right)^2 = \left(F_{(m+1)(m+1)} - \sum_{j=1}^{m}{F_{jm}F_{(m+1)(m+1)}}\right)
\]
By Lemma \ref{lem:buildingblock}, 
$F_{jm}F_{(m+1)(m+1)} + \sum_{k \in [m] \setminus \{j\}}{T_{jmk}} = F_{jm} - F_{j(m+1)} $. This implies that \begin{align*}
-\sum_{m=1}^{n-1}{\sum_{j=1}^{m}{\left(F_{jm}F_{(m+1)(m+1)} + \sum_{k \in [m] \setminus \{j\}}{T_{jmk}} \right)}}&=
-\sum_{j=1}^{n-1}{\sum_{m = j}^{n-1}{\left(F_{jm} - F_{j(m+1)} \right)}} \\
&= \sum_{j=1}^{n-1}{(F_{jn} - F_{jj})}
\end{align*}
Now observe that 
$\sum_{m=1}^{n-1}{F_{(m+1)(m+1)}} = \sum_{j = 2}^{n-1}{F_{jj}} + F_{nn} = -1 + \sum_{j=1}^{n-1}{F_{jj}} + F_{nn}$. Putting everything together, the equality holds. Since each $T_{jmk}$ is a weakening of a transitivity axiom and each $F_{jn}$ is a non-minimality axiom, this indeed gives a sum of squares proof that these axioms are unsatisfiable.
\end{proof}
\section{Open Problems}
Our work raises a number of open problems. For the pigeonhole principle, while we have proved an exponential total coefficient size lower bound on Nullstellensatz proofs, there is a lot of room for further work. Some questions are as follows.
\begin{enumerate}
    \item For the pigeonhole principle, our lower bound is $2^{\Omega(n)}$ while the trivial upper bound is $O(n!)$. Can we improve the lower and/or upper bound?
    \item If we increase the number of pigeons from $n$ to $n+1$ while still having $n - 1$ holes, our lower bound proof no longer applies. Can we prove a total coefficient size lower bound on Nullstellensatz when there are $n+1$ or more pigeons? How does the minimum total coefficient size of a proof depend on the number of pigeons?
    \item Can we show total coefficient size lower bounds for resolution-like proofs of the pigeonhole principle?
    \item How much of an effect would adding the axioms that pigeons can only go to one hole have on the minimum total coefficient size needed to prove the pigeonhole principle?
\end{enumerate}
We are still far from understanding the total coefficient size of proofs for the ordering principle. Two natural questions are as follows.
\begin{enumerate}
    \item Can we prove superpolynomial lower bounds on the total coefficient size of Nullstellensatz proofs for the ordering principle and/or improve the $O(2^n)$ upper bound?
    \item Are there resolution-like proofs for the ordering principle with polynomial total coefficient size? If so, this shows that the seemingly dynamic $O(n^3)$ size resolution proof of the ordering principle \cite{staalmarck1996short} can be captured by a one line resolution-like proof. If not, this gives a natural example separating resolution proof size and the total coefficient size of resolution-like proofs.  
\end{enumerate}
Finally, we can ask what relationships and separations we can show between all of these different proof systems. Some questions are as follows.
\begin{enumerate}
    \item Are there natural examples where the minimum total coefficient size is very different (either larger or smaller) than the minimum size for Nullstellensatz, resolution-like, or sum of squares proofs?
    \item Can the minimum total coefficient size of a strong proof system be used to lower bound the size of another proof system? For example, can resolution proof size be lower bounded by the minimum total coefficient size of a sum of squares proof or can we find an example where there is a polynomial size resolution proof but any sum of squares proof has superpolynomial total coefficient size?
\end{enumerate}
\bibliographystyle{alpha}
\bibliography{main}
\end{document}